\documentclass[conference]{IEEEtran}
\IEEEoverridecommandlockouts
\usepackage{cite}
\usepackage{amsmath,amssymb,amsfonts}
\usepackage{algorithmic}
\usepackage{graphicx}
\usepackage{textcomp}
\usepackage{xcolor}
\def\BibTeX{{\rm B\kern-.05em{\sc i\kern-.025em b}\kern-.08em
    T\kern-.1667em\lower.7ex\hbox{E}\kern-.125emX}}

\usepackage{amsthm}
\usepackage{amssymb}    
\usepackage{stmaryrd}
\usepackage{bm}
\usepackage{mathrsfs}
\usepackage{textcomp} 

\newtheorem{theorem}{Theorem}
\newtheorem{lemma}{Lemma}

\newtheorem{definition}{Definition}
	\theoremstyle{definition}
\newtheorem{example}{Example}
	\theoremstyle{remark}

\makeatletter  
\def\th@plain{%
  \thm@notefont{}
  \itshape 
}
\def\th@definition{%
  \thm@notefont{}
  \normalfont 
}
\makeatother

\newcommand{\weight}{\vartheta}
\newcommand{\Bigcup}{\text{\Large{$\Cup$}}}

\newcommand{\one}{\texttt{\bfseries one}}
\newcommand{\skis}{\texttt{\bfseries skis}}

\usepackage{color}

\begin{document}

\title{Kleene Algebra With Tests for Weighted Programs\thanks{This work was supported by the grant GA22-16111S of the Czech Science Foundation. The author is grateful to the reviewers for useful comments.}}

\author{\IEEEauthorblockN{Igor Sedl\'{a}r}
\IEEEauthorblockA{\textit{The Czech Academy of Sciences, Institute of Computer Science} \\
Prague, The Czech Republic \\
sedlar@cs.cas.cz}
}

\maketitle

\begin{abstract}
Weighted programs generalize probabilistic programs and offer a framework for specifying and encoding mathematical models by means of an algorithmic representation. Kleene algebra with tests is an algebraic formalism based on regular expressions with applications in proving program equivalence. We extend the language of Kleene algebra with tests so that it is sufficient to formalize reasoning about a simplified version weighted programs. We introduce relational semantics for the extended language, and we generalize the relational semantics to an appropriate extension of Kleene algebra with tests, called Kleene algebra with weights and tests. We demonstrate by means of an example that Kleene algebra with weights and tests offers a simple algebraic framework for reasoning about equivalence and optimal runs of weighted programs.
\end{abstract}


\section{Introduction}

\emph{Weighted programs} \cite{BatzEtAl2022} add two features to standard while programs \cite{AptEtAl2009}: (i) nondeterministic branching, and (ii) the ability to weight the current execution trace. As such, they generalize probabilistic programs \cite{BartheEtAl2020}, in which execution can branch based on the outcome of a random coin flip.

Batz et al.~\cite{BatzEtAl2022} argue that, in addition to being an interesting generalization, weighted programs offer a useful programming paradigm for specifying and encoding mathematical models, such as optimization problems or probability distributions, by means of an algorithmic representation. They introduce structural operational semantics and a Dijkstra-style weakest precondition calculus for weighted programs, and they use this calculus to determine competitive ratios of weighted programs.

Kleene algebra with tests \cite{Kozen1997}, \textsf{KAT}, is an extension of the algebraic theory of regular languages that provides a simple algebraic formalism for equational specification and verification of while programs, and for proving equivalence of programs. In this note we extend the language of \textsf{KAT} so that it is sufficient to formalize reasoning about a propositional abstraction of weighted programs where assignments are replaced by unstructured atomic actions. We introduce relational semantics for the extended language, and we generalize the relational semantics to an appropriate extension of Kleene algebra with tests, called Kleene algebra with weights and tests, or \textsf{KAWT}. We demonstrate by means of an example that \textsf{KAWT} offers a simple algebraic framework for reasoning about equivalence and optimal runs of weighted programs.

The note is structured as follows. Section \ref{sec:weighted} introduces \emph{weighted regular programs}, an expansion of regular programs with weights, and shows that weighted regular programs can express a propositional abstraction of weighted programs where assignments are replaced by unstructured atomic actions. Section \ref{sec:relational} introduces a relational semantics for weighted regular programs using weighted transition systems where weights are elements of an arbitrary semiring. Section \ref{sec:KA} recalls the basics of Kleene algebra and Section \ref{sec:KAWT} introduces an expansion of Kleene algebra with tests that generalizes the relational semantics of weighted regular programs. The expansion is called \emph{Kleene algebra with weights and tests}, and it adds a semiring subalgebra (representing weights) to Kleene algebra with tests. Section \ref{sec:using} demonstrates by means of an example how Kleene algebra with weights and tests can be used to reason about equivalence and optimal runs of regular while programs. 

\section{Weighted programs}\label{sec:weighted}

It is well-known that while programs can be represented using the syntax of \emph{regular programs} \cite{Pratt1976,FischerLadner1979}. The latter contains nondeterministic branching (or choice) as one of the basic operations, so it is natural to extend regular programs with weights to capture weighted programs. For the sake of simplicity, we will work with a propositional abstraction of weighted programs where assignments $\mathtt{x} := E$ are replaced by unstructured ``atomic programs'' in the style of \cite{FischerLadner1979,Kozen1997}. We will see that even in this simpler setting we are able to carry out a significant amount of reasoning about weighted programs. 

\begin{definition}
A \emph{signature} is $\Sigma = (\mathsf{P}, \mathsf{B}, \mathsf{F})$, a triple of disjoint sets of variables, intuitively representing atomic programs, Boolean tests, and weighted tests, respectively. The set of \emph{weighted regular programs over $\Sigma$}, $\mathsf{WRP}_{\Sigma}$, consists of expressions of the following sorts:
\begin{center}
\begin{tabular}{lrl}
Boolean expressions \quad & $b,c :=$ & $\mathtt{b} \mid \mathtt{0} \mid \mathtt{1} \mid b \cdot c \mid b + c \mid \neg b$\\
Weightings & $f, g := $ & $\mathtt{f} \mid \mathtt{0} \mid \mathtt{1} \mid f \cdot g \mid f + g$\\
Programs & $p, q :=$ & $\mathtt{p} \mid b \mid f \mid p + q \mid p \cdot q \mid p^{*}$
\end{tabular}
\end{center}
\end{definition}
\noindent (It is assumed that $\mathtt{b} \in \mathsf{B}$, $\mathtt{f} \in \mathsf{F}$, and $\mathtt{p} \in \mathsf{P}$.) We will sometimes write $\bar{b}$ instead of $\neg b$ and $pq$ instead of $p \cdot q$. 

The language of weighted regular programs extends the language of Kleene algebra with tests \cite{Kozen1997} corresponding to regular programs by adding  the sort of ``weightings'', that is, semiring terms representing assignments of weights to computation paths. Weights are represented by elements of an abstract semiring \cite{DrosteEtAl2009,KuichSalomaa1986}. Recall how the language of \textsf{KAT}  expresses standard control flow commands of while programs:
\begin{itemize}
\item \textbf{skip} $:= \mathtt{1}$ and \textbf{abort} $:= \mathtt{0}$
 \item sequential composition $p \,;\, q := p q$
 \item \textbf{if} $b$ \textbf{then} $p$ \textbf{else} $q$ $:= (bp) + (\neg bq)$
 \item \textbf{while} $b$ \textbf{do} $p$ $:= (bp)^{*}\neg b$ 
 \end{itemize}
Note that nondeterministic branching (choice) is represented in the language of regular programs by the $+$ operator.
 
 \begin{example}\label{exam:ski}
 Recall the ``ski rental program'', the main motivating example of a weighted program in \cite{BatzEtAl2022}:\footnote{The ski rental problem is an optimization problem regarding a situation where one goes for a skying trip for $n$ days and has the choice of renting a pair of skis for $1$ Euro per day vs.~buying a pair of skis for $y$ Euros.}

\begin{figure}[h]
\begin{algorithmic}[1]
\WHILE{$n > 0$}
\STATE $n := n-1$ $;$
\STATE $\big \{$ $ \odot 1 $
\STATE $\oplus$
\STATE $ \odot y \, ; \, n := 0 $ $\big \}$
\ENDWHILE
\end{algorithmic}
\end{figure}
 
 The operator $\oplus$ in line 4 expresses nondeterministic branching between computation sequences, one of which executes the ``weighting'' $\odot 1$, or ``add one unit of weight'' and returns to the beginning of the while loop, and the other executes $\odot y$, or ``add $y$ units of weight'', then assigns $n := 0$, and then returns to the beginning of the loop, after which the Boolean test $n > 0$ is evaluated as False and the computation halts.
 
The structure of the ski rental program can be represented by the following weighted regular program (for the sake of readability, we use descriptive variable names, bold font for weighting variables, and we enclose Boolean expressions in curly brackets):
\begin{equation}\label{eq:ski}
\Big( \{ \texttt{neq0} \} \big( \texttt{sub1} (\one + \skis \, \texttt{end}) \big) \Big)^{*}\, \{ \neg \texttt{neq0} \} \, .
\end{equation}
In this program, $\texttt{neq0}  \in \mathsf{B}$ represents the test $n > 0$, $\texttt{sub1} \in \mathsf{P}$ represents the instruction $n := n - 1$ to subtract $1$ from the value of $n$, $\texttt{end}  \in \mathsf{P}$ represents the assignment $n := 0$ that ends the loop, and $\one, \skis \in \mathsf{F}$ represent the weightings $\odot 1$ (adding one unit of weight) and $\odot y$ (adding the price of the skis), respectively.  
 \end{example}
 

\section{Relational semantics}\label{sec:relational}

In this section we introduce relational semantics for weighted regular programs, based on semiring-valued transition systems.

\begin{definition}
A \emph{semiring} is $\mathbf{S} = (S, +, \cdot, 1, 0)$ where
\begin{itemize}
\item $(S, +, 0)$ is a commutative monoid;
\item $(S, \cdot, 1)$ is a monoid;
\item $x \cdot (y + z) = (x \cdot y) + (x \cdot z)$ and $(x + y) \cdot z = (x \cdot z) + (y \cdot z)$;
\item $0 \cdot x = 0 = x \cdot 0$.
\end{itemize}
A semiring is \emph{idempotent} iff $x + x = x$ for all $x \in S$. A semiring is \emph{complete} iff $(S, +, 0)$ is a complete monoid and the following distributivity laws hold:
\begin{equation*}
\sum_{i \in I} (x \cdot x_i) = x \cdot \left ( \sum_{i \in I} x_i \right ) \qquad
\sum_{i \in I} (x_i \cdot x) =  \left ( \sum_{i \in I} x_i \right ) \cdot x
\end{equation*}
\end{definition}
(Complete idempotent semirings are also known as \emph{quantales}.) We define the \emph{natural order} $\preceq$ on a semiring as follows:
$$x \preceq y \iff \exists z (x + z = y) \, .$$
We note that in idempotent semirings the natural order coincides with the semilattice order $\leq$ defined by $x \leq y \iff x + y = y$.

\begin{example}
(1) An example of a complete idempotent semiring that is well known especially from shortest path algorithms is the \emph{tropical semiring} over extended natural numbers
\begin{equation*}
\mathbf{T} = (\mathbb{N}^{\infty}, \mathrm{min}, +, 0^{\mathbb{N}}, \infty) 
\end{equation*}
where
\begin{itemize}
 \item $\mathbb{N}^{\infty} = \mathbb{N} \cup \{ \infty \}$, where $\infty \notin \mathbb{N}$;
 \item $\mathrm{min}$ is the minimum operation extended to $\mathbb{N}^{\infty}$ by defining $\mathrm{min}(n, \infty) = \mathrm{min}(\infty, n) = n$ for all $n$; $\mathrm{min}$ is seen as semiring addition (hence, $\infty$ is the minimal element in the ordering $\leq$ defined by $x \leq y$ iff $\mathrm{min}(x,y) = y$);
 \item $+$, representing semiring multiplication, is addition on $\mathbb{N}$  and $n + \infty = \infty = \infty + n$ for all $n$ (hence, $\infty$ is the annihilator element);
 \item $0^{\mathbb{N}}$, the natural number zero, is the neutral element with respect to $+$.
 \end{itemize}
 
 (2) Another example well-known from the theory of fuzzy logic is the \emph{{\L}ukasiewicz semiring} $\textbf{{\L}} = ([0,1], \mathrm{max}, \otimes, 1, 0)$ where $[0,1]$ is the real unit interval and $\otimes$ is the {\L}ukasiewicz t-norm $x \otimes y = \mathrm{max}\{ 0, x + y - 1 \}$.
\end{example}

Intuitively, semirings can be seen as representing a set of \emph{weights} together with an operation of weight addition (semiring multiplication $\cdot$), a weight comparison relation $\preceq$ induced by semiring addition $+$, and two designated weights, namely, ``no weight'' $1$ and the ``absolute'' weight $0$.

\begin{definition}
Let $\mathbf{S}$ be a semiring. An \emph{$\mathbf{S}^\Sigma$-transition system} is $(X, L)$ where $X \neq \emptyset$ and $L$ is a function such that
\begin{center}
$L (\mathtt{p}) \subseteq X \times X$ \qquad
$L (\mathtt{b}) \subseteq X$\qquad
$L(\mathtt{f}) \in \mathbf{S} \, .$
\end{center}
\end{definition}
\noindent Intuitively, $(s,t) \in L(\mathtt{p})$ means that the atomic program $\mathtt{p}$ may terminate in state $t$ when executed in state $s$; $s \in L(\mathtt{b})$ means that the atomic Boolean expression $\mathtt{b}$ is evaluated to True in $s$; and $L(\mathtt{f}) \in \mathbf{S}$ is the weight assigned to the atomic weighting expression $\mathtt{f}$.

Equivalently, $L$ in an $\mathbf{S}^{\Sigma}$-transition system can be seen as a function from $\bigcup \Sigma$ to $\mathbf{S}^{X \times X}$ such that
\begin{itemize}
\item $L(\mathtt{p}), L(\mathtt{b}) \in \{ 0, 1 \}^{X \times X}$;
\item $L(\mathtt{b})(s,t) = 0$ and $L(\mathtt{f})(s,t) = 0$ if $s \neq t$;
\item $L(\mathtt{f})(s,t) = L(\mathtt{f})(s', t')$ for all $s,s',t,t' \in X$.
\end{itemize}
\noindent (In effect, $L(\mathtt{b})$ and $L(\mathtt{f})$ are functions from the identity relation on $S$ which in turn represents $S$.) 

In a given $\mathbf{S}^{\Sigma}$-semiring, $L$ can be extended to a function $\mathsf{WRP}_{\Sigma} \to \mathbf{2}^{X \times X}$ specifying the interpretation of each weighted regular program by lifting the weighted regular program operations to the set of functions $X \times X \to \mathbf{S}$ :
\begin{definition}
Let $\lambda, \lambda' \in \mathbf{S}^{X \times X}$ and let $\theta \in \{ 0, 1 \}^{X \times X}$  
\begin{itemize}
\item $\mathbf{1} = \mathrm{id}_X$ (the identity relation on $X$);
\item $\mathbf{0}(s,t) = 0$;
\item $(\lambda + \lambda')(s,t) = \lambda (s,t) + \lambda' (s,t)$;
\item $(\lambda \cdot \lambda')(s,t) = \sum \{ \lambda (s,u) \cdot \lambda' (u,t) \mid u \in X \}$;
\item $(\neg \theta) (s,t) = 
		\begin{cases}
		1 & \text{if } s = t \text{ and } \theta (s,t) = 0\\
		0 & \text{otherwise;}
		\end{cases}
		$
\item $\lambda^{0} = \mathbf{1}$ and $\lambda^{n+1} = \lambda^{n} \cdot \lambda$.
\end{itemize}
If $\mathbf{S}$ is a complete semiring, then we define:
\begin{itemize}
\item $\lambda^{*} = \sum_{n \geq 0} \lambda^{n}$.		
\end{itemize}
\end{definition}

It is easily checked that $(\mathbf{S}^{X \times X}, \cdot, +, \mathbf{1}, \mathbf{0})$ is a semiring; see Lemma \ref{lem:lifting} below. We denote $\{ 0, 1 \}^{X \times X}$ as $2^{X \times X}$.

\section{Kleene algebra}\label{sec:KA}

Kleene algebras \cite{Kozen1990,Kozen1994} are structures that arise naturally in the study of regular languages, finite automata and shortest path algorithms, for instance. Kleene algebra offers an elegant  framework for equational reasoning about regular programs \cite{Kozen1997}. In this section we recall the basic preliminaries on Kleene algebra.

\begin{definition}
A \emph{Kleene algebra} is an idempotent semiring with a unary operation $^{*}$ that satisfies the following (quasi)equations:
\begin{gather*}
1 + xx^{*} \leq x^{*} \qquad
1 + x^{*} x \leq x^{*}\\
y + xz \leq z \to x^{*} y \leq z \qquad
y + zx \leq y \to y x^{*} \leq z
\end{gather*}
A Kleene algebra is \emph{$^{*}$-continuous} iff it satisfies
\begin{equation*}
x y^{*} z = \sum_{n \geq 0} x y^{n} z \, .
\end{equation*}
\end{definition}

\begin{example}
(1) The Kleene algebra of binary relations on a non-empty set $X$ is
\begin{equation*}
\mathbf{Rel}_{X} = (\mathbf{2}^{X \times X}, \cup, \circ, \,^{*}, \mathrm{id}_X, \emptyset) \, ,
\end{equation*}
where $\mathbf{2}^{X \times X}$ is the set of all binary relations on $X$, $\cup$ is union, $\circ$ is relational composition, $\,^{*}$ is reflexive transitive closure, and $\mathrm{id}_X$ is the identity relation on $X$.

(2) Take a finite set of symbols $\Delta$ and let $\Delta^{*}$ be the set of finite sequences over $\Delta$, including the empty sequence $\epsilon$. Let $\mathrm{Reg}_{\Delta}$ be the smallest subset of $\mathbf{2}^{\Delta^{*}}$  (i.e.~the set of all languages over $\Delta$) that contains the empty set $\emptyset$, the set $\{ \epsilon \}$, and $\{ \mathtt{x} \}$ for all $\mathtt{x} \in \Delta$ that is closed under finite unions and the following operations:
\begin{itemize}
\item $X \cdot Y = \{ xy \mid x \in X \And y \in Y \}$, i.e.~the set of concatenations of strings from $X$ with strings from $Y$;
\item $X^{*} = \bigcup_{n \geq 0} X^{n}$, where $X^{0} = \{ \epsilon \}$ and $X^{n+1} = X^{n} \cdot X$, i.e.~the set of all strings that can be parsed as concatenations of a finite number of strings in $X$ (Kleene iteration).
\end{itemize}
The set $\mathrm{Reg}_{\Delta}$ is called the set of regular languages over $\Delta$. The Kleene algebra of regular languages over $\Delta$ is
\begin{equation*}
\mathbf{Reg}_{\Delta} = (\mathrm{Reg}_{\Delta}, \cup, \cdot, \,^{*}, \{ \epsilon \}, \emptyset ) \, .
\end{equation*}
\end{example}

\begin{definition}
A \emph{Kleene algebra with tests} \cite{Kozen1997} is a structure of the form
 $$\mathbf{K} = (K, B, \cdot, +, \,^{*}, \,^{-}, 1, 0)$$ where
 \begin{itemize}
 \item $(K, \cdot, +, \,^{*}, 1, 0)$ is a Kleene algebra;
 \item $B \subseteq K$, and $(B, \cdot, +, \,^{-}, 1, 0)$ is a Boolean algebra.
 \end{itemize}
A Kleene algebra with tests is $^{*}$-continuous iff its underlying Kleene algebra is $^{*}$-continuous.
\end{definition}

In a Kleene algebra with tests, elements of $B$ are seen as ``tests'' of statements formulated in a Boolean language, while elements of $K$ in general represent ``structured actions''.

\begin{example}
(1) Every Kleene algebra is a Kleene algebra with tests. Take $B = \{ 1, 0 \}$ and define $\,^{-}$ as Boolean complementation on $B$.

(2) A relational Kleene algebra with tests over some set $X$ is the expansion of $\mathbf{Rel}_{X}$ with the set of subsets of $\mathrm{id}_X$ (seen as $B$), and the complementation operation on the set of subsets of $\mathrm{id}_X$ (seen as $\,^{-}$).

(3) Take two finite sets $\mathsf{A}$ and $\mathsf{T}$ of program and Boolean variables, respectively. We assume that $\mathsf{A} = \{ \mathtt{b}_1, \ldots, \mathtt{b}_n \}$ is ordered in some fixed but arbitrary way. An \emph{atom} over $\mathsf{A}$ is a sequence $\mathtt{b}_1^{\pm} \ldots \mathtt{b}_n^{\pm}$, where $\mathtt{b}_i^{\pm} \in \{ \mathtt{b}_i, \bar{\mathtt{b}}_i \}$. Let $1_\mathsf{A}$ be the set of all atoms over $\mathsf{A}$. A \emph{guarded string} over $\mathsf{A}, \mathsf{T}$ is any sequence of the form (for $k \geq 0$)
 \begin{equation*}
 A_0 p_1 A_1 p_2 A_2 \ldots p_k A_{k} 
 \end{equation*}
 where each $A_i$ is an atom over $\mathsf{A}$ and each $p_i \in \mathsf{T}$. Guarded strings can be seen as representations of execution traces of programs (abstract states are replaced by atoms). Let $2^{GS_{\mathsf{A}, \mathsf{T}}}$ be the set of all sets of guarded strings over $\mathsf{A}, \mathsf{T}$. Obviously $1_{\mathsf{A}} \subseteq 2^{GS_{\mathsf{A}, \mathsf{T}}}$. The \emph{coalesced product} operation $\diamond$ is a partial binary function on $GS_{\mathsf{A}, \mathsf{T}}$ defined as follows:
 \begin{equation*}
 x A \diamond A' y =
 \begin{cases}
 x A y & \text{if } A = A'\\
 \text{undefined} & \text{otherwise}.
 \end{cases}
 \end{equation*}
 The coalesced product operation is lifted to sets of guarded strings in an obvious way. Note that the coalesced product operation on sets of guarded strings is a total function. The \emph{algebra of guarded languages} over $\mathsf{A}, \mathsf{T}$ \cite{KozenSmith1997},
\begin{equation*}
 \mathbf{G}_{\mathsf{A}, \mathsf{T}} = (2^{GS_{\mathsf{A}, \mathsf{T}}}, 2^{1_\mathsf{A}}, \cup, \diamond, \,^{*}, \,^{-}, 1_\mathsf{A}, \emptyset) \, ,
 \end{equation*} 
  is a $^{*}$-continuous Kleene algebra with tests ($X^{*} = \bigcup_{ n \geq 0} X^{n}$, where exponentiation is defined using $\diamond$).
\end{example}

\begin{lemma}\label{lem:lifting}
Let $X$ be a non-empty set and $\mathbf{S}$ a semiring. Then
\begin{enumerate}
\item $\mathbf{S}(X) = (\mathbf{S}^{X \times X}, \cdot, +, \mathbf{1}, \mathbf{0})$ is a semiring;
\item if $\mathbf{S}$ is idempotent (complete), then $\mathbf{S}(X)$ is idempotent (complete);
\item if $\mathbf{S}$ is idempotent and complete, then $\mathbf{S^{*}}(X) = (\mathbf{S}(X), \,^{*})$ is a $^{*}$-continuous Kleene algebra;
\item if $\mathbf{S}$ is idempotent and complete, then $\mathbf{S^{*}_{2}}(X) = (\mathbf{S^{*}}(X), \mathbf{2}^{\mathrm{id}_X}, \neg )$ is a $^{*}$-continuous Kleene algebra with tests.
\end{enumerate}
\end{lemma}

(We note that $\mathbf{S^{*}_{2}}(X)$ provides a natural example of a Kleene algebra with tests where the set of tests, that is $ \mathbf{2}^{\mathrm{id}_X}$, is not identical to the set of elements under $\mathbf{1} = \mathrm{id}_X$, where $\lambda \leq \lambda'$ iff $\lambda(s,t) \leq \lambda'(s,t)$ for all $s,t \in X \times X$.)

\begin{lemma}\label{lem:constant}
The subalgebra of $\mathbf{S}(X)$ consisting of all constant functions is isomorphic to $\mathbf{S}$.
\end{lemma} 

\section{Kleene algebra with weights and tests}\label{sec:KAWT}
 
 Given the applications of Kleene algebra in reasoning about regular and while programs, it is natural to consider a Kleene-algebraic perspective on weighted programs. Gomes et al.~\cite{GomesEtAl2019} formulate a generalization of \textsf{KAT} called \emph{graded} \textsf{KAT} (or \textsf{GKAT}) where the Boolean algebra of tests is replaced by a more general algebraic structure. Batz et al.~\cite{BatzEtAl2022} point out that a deeper study of the applicability of \textsf{GKAT} to reasoning about weighted programs is an interesting problem to look at. Weighted programs in the sense of \cite{BatzEtAl2022} combine Boolean tests and weightings (weighted tests), and so it seems natural to consider a generalization of the \textsf{GKAT} approach.
 
 \begin{definition}
 A \emph{Kleene algebra with weights and tests} is a structure
 $$\mathbf{K} = (K, B, S, \cdot, +, \,^{*}, \,^{-}, 1, 0)$$ where
 \begin{itemize}
 \item $(K, \cdot, +, \,^{*}, 1, 0)$ is a Kleene algebra;
 \item $B \subseteq K$ and $S \subseteq K$;
 \item $(B, \cdot, +, \,^{-}, 1, 0)$ is a Boolean algebra;
 \item $(S, \cdot, +, 1, 0)$ is a semiring.
 \end{itemize}
 A \emph{valuation in $\mathbf{K}$} is any homomorphism $v$ from $\mathsf{WRP}_{\Sigma}$ to $\mathbf{K}$ such that $v(\mathtt{b}) \in B$ and $v(\mathtt{f}) \in S$. Two programs $p, q \in \mathsf{WRP}_{\Sigma}$ are \emph{equivalent} in a class $\mathcal{K}$ of Kleene algebras with weights and tests iff $v(p) = v(q)$ for all valuations $v$ in all $\mathbf{K} \in \mathcal{K}$ .
 \end{definition}
 
 Clearly each KAT is a KAWT; just take any subuniverse $S$ of $K$. However, there are more interesting examples.
 
 \begin{example}

(1) Our first example combines a Kleene algebra of guarded languages with the tropical semiring. Take the set of guarded strings over some fixed $\mathsf{A}$ and $\mathsf{T}$, and denote it as $GS$. A function $\lambda: GS \to \mathbb{N}^{\infty}$, assigning a weight to each guarded string, can be seen as specifying the weights of execution traces of a program (if $\lambda(s) = \infty$, then either $s$ is not a trace that can be generated by the program corresponding to $\lambda$ or the trace carries an ``absolute weight''). Note that functions $\lambda: GS \to \mathbb{N}^{\infty}$ generalize sets of guarded strings ($GS \to \{ 0^{\mathbb{N}}, \infty \}$), i.e.~elements of the Kleene algebra $\mathbf{G}$. Take
\begin{equation*}
\mathbf{T^{G}} = ((\mathbb{N}^{\infty})^{GS}, B, S, \cdot, +, \,^{*}, \,^{-}, 1, 0)
\end{equation*}
where $(\mathbb{N}^{\infty})^{GS}$ is the set of all functions from $GS$ to $\mathbb{N}^{\infty}$ and
\begin{itemize}
\item $1 (s) =
\begin{cases}
0^{\mathbb{N}} & \text{if } s \in 1_{\mathsf{A}}\\
\infty & \text{otherwise;}
\end{cases}
$

\item $0 (s) = \infty$ for all $s$;

\item $\lambda \in B$ iff $\lambda(s) \in \{ 0^{\mathbb{N}}, \infty \}$ for all $s \in GT$ and $\lambda(s) = 0^{\mathbb{N}}$ only if $s \in 1_{\mathsf{A}}$; 

\item $S = \{ \lambda \mid \lambda(s) = \infty$ if $s \not\in 1_{\mathsf{A}}$ and $\lambda (A_i) = \lambda(A_j)$ for all $A_i, A_j \in 1_{\mathsf{A}} \}$;

\item $(\lambda \cdot \lambda')(s) = \mathrm{min} \{ \lambda (t) +^{\mathbb{N}^{\infty}} \lambda'(u) \mid s = t \diamond u \}$;

\item $(\lambda + \lambda')(s) = \mathrm{min} \{ \lambda(s), \lambda'(s) \}$;

\item $(\lambda^{*})(s) = \mathrm{min}_{n \geq 0} \big ( \lambda^{n}(s) \big )$ where $\lambda^{n}$ is defined using $\cdot\, ;$ 

\item $\,^{-} : B \to B$ such that $(\bar{\lambda})(s) = \infty$ iff $\lambda(s) \neq \infty$.
\end{itemize}
It can be shown that $\mathbf{T^{G}}$ is a ($^{*}$-continuous) Kleene algebra with weights and tests where $(S, \cdot, +, 1, 0) \cong \mathbf{T}$. $\mathbf{T^{G}}$ also contains a (proper) subalgebra isomorphic to $\mathbf{G}$, namely, $( \{ 0^{\mathbb{N}}, \infty \}^{G}, B, \cdot, +, \,^{*}, \,^{-}, 1, 0 )$.

(2) A similar construction can be carried out using $\mathbf{T}$ and $\mathbf{Rel}_X$. We leave the details to the reader, noting just that the coalesced product operation $\diamond$ can be defined on $X \times X$ by stipulating that $(s,t) \diamond (u,v) $ is $ (s,v)$ if $t = u$ and is undefined otherwise.
 \end{example}
 
 The construction encountered in the previous example can obviously be carried out in a more general setting. In fact, we can use a slight generalization of the well known notion of \emph{formal power series} coming up for instance in the study of weighted automata \cite{DrosteEtAl2009,KuichSalomaa1986}. Formal power series are functions from a monoid to a semiring; we will build on functions from certain partial semigroups.
 
%
%

\begin{definition}
A \emph{partial semigroup with identity} is $(G, D, \diamond, I)$ where $G \neq \emptyset$, $D$ is a binary relation on $G$, $\diamond : D \to G$, and $I \subseteq G$ such that
\begin{itemize}
\item $D(x,y) \And D(x \diamond y, z)$ iff $D(y,z) \And D(x, y \diamond z)$;
\item $D(x,y) \And D(x \diamond y, z)$ only if $(x \diamond y) \diamond z = x \diamond (y \diamond z)$;
\item for all $x$ there is $y$: $I(y)$ and $D(x,y)$;
\item for all $x$ there is $y$: $I(y)$ and $D(y,x)$;
\item $D(x,y)$ and $I(y)$ only if $x = x \diamond y$;
\item $D(y,x)$ and $I(y)$ only if $x = y \diamond x$.
\end{itemize}
\end{definition}
We note that the first two conditions above are stronger than the standard definition of a partial semigroup \cite{GudderSchelp1970}. We'll often write $xy$ instead of $x \diamond y$.
 
 \begin{example}\label{exam:partial}
 (1) An example of a partial semigroup with identity is $ \mathbf{Cart}_X = (X \times X, D, \diamond, \mathrm{id}_X)$ where $X \times X$ is the Cartesian product on a non-empty set $X$, $((x_1, y_1), (x_2, y_2)) \in D$ iff $y_1 = x_2$, $(x, y) \diamond (y, z) = (x,z)$, and $\mathrm{id}_X = \{ (x,x) \mid x \in X \}$ is the identity relation on $X$. 
 
 (2) $ \mathbf{Gu}_{\mathsf{A, T}} = (GS, D, \diamond, 1_{\mathsf{A}})$ where $GS$ is the set of guarded strings over some $\mathsf{A}$ and $\mathsf{T}$, $\diamond$ is the coalesced product operation, $D$ is the set of pairs $(s,t)$ such that $s \diamond t$ is defined, and $1_{\mathsf{A}}$ is the set of atoms over $\mathsf{A}$. 
 
 (3) $ \mathbf{Str}_{\Delta} = (\Delta^{*}, D, \,^{\frown}, \epsilon)$ where $\Delta^{*}$ is the set of finite sequences (strings) over a finite alphabet $\Delta$ (including the empty string $\epsilon$), $D$ is the universal relation on $\Delta^{*}$ and $\,^{\frown}$ is the concatenation operation (which is actually a total operation on $\Delta^{*}$). 
 \end{example}

Many well-known examples of Kleene algebras (with tests) are in fact algebras of functions from specific (regular) partial semigroups to the two-element Boolean semiring $\mathbf{Bo} = (\{ 1, 0 \}, \land, \lor, 1, 0)$; the reader can easily verify this by considering some of the previous examples.


\begin{definition}
Let $\mathbf{P} = (G, D, \diamond, I)$ be a partial semigroup with identity, and let $\mathbf{S} = (S, \cdot^{\mathbf{S}}, +^{\mathbf{S}}, 1^{\mathbf{S}}, 0^{\mathbf{S}})$ be a complete idempotent semiring. We define
\begin{equation*}
\mathbf{S^{P}} = (S^{G}, B, W, \cdot, +, \,^{*}, \,^{-}, 1, 0 )
\end{equation*} where $S^{G}$ is the set of all functions form $G$ to $S$ and
\begin{itemize}
\item $1 (x) =
\begin{cases}
1^{\mathbf{S}} & \text{if } x \in I\\
0^{\mathbf{S}} & \text{otherwise;}
\end{cases}
$

\item $0 (x) = 0^{\mathbf{S}}$ for all $s$;

\item $\lambda \in B$ iff $\lambda(x) \in \{ 1^{\mathbf{S}}, 0^{\mathbf{S}} \}$ for all $x \in G$ and $\lambda(x) = 1^{\mathbf{S}}$ only if $x \in I$; 

\item $W = \{ \lambda \mid \lambda(x) = 0^{\mathbf{S}}$ if $x \not\in I$ and $\lambda (y) = \lambda(z)$ for all $y, z \in I \}$;

\item $(\lambda \cdot \lambda')(x) = \sum \{ \lambda (y) \cdot^{\mathbf{S}} \lambda'(z) \mid (y,z) \in D \And y \diamond z = x \}$ (the sum uses $+^{\mathbf{S}}$);

\item $(\lambda + \lambda')(x) = \lambda(x) +^{\mathbf{S}} \lambda'(x)$;

\item $(\lambda^{*})(x) = \sum_{n \geq 0} \big ( \lambda^{n}(x) \big )$  where $\lambda^{n}$ is defined using $\cdot\, ;$ 

\item $\,^{-} : B \to B$ such that $(\bar{\lambda})(x) = 1^{\mathbf{S}}$ iff $\lambda(x) \neq 1^{\mathbf{S}}$.
\end{itemize}
\end{definition}

\begin{theorem}\label{thm:Partial_to_KAWT}
If $\mathbf{P}$ is a partial semigroup with identity and $\mathbf{S}$ is a complete idempotent semiring, then $\mathbf{S^P}$ is a $^{*}$-continuous Kleene algebra with weights and tests where $$(W, \cdot, +, 1, 0) \cong \mathbf{S} \, .$$
\end{theorem}

Given the relation of Kleene algebra to the algebra of regular languages \cite{Kozen1994}, and the relation of Kleene algebra with tests to the algebra of regular guarded languages \cite{KozenSmith1997}, it is natural to ask if, for any given class of semirings $\mathcal{S}$ (closed under isomorphisms) there is a Kleene algebra with weights and tests whose elements are sets of strings over some alphabet (a ``language-theoretic'' algebra, to use the terminology of \cite{KozenSmith1997}) which is free in the class of Kleene algebras with weights and tests whose semiring of weights belongs to $\mathcal{S}$. (That is, if the equational theory of the given class of Kleene algebras with weights and tests is complete with respect to a class language-theoretic algebras.) 

We will not answer this question here, but we will note that $\mathbf{T^{G}}$ over any $\mathsf{A}, \mathsf{T}$ (the algebra of functions from guarded strings over $\mathsf{A}, \mathsf{T}$ to extended natural numbers) is isomorphic to a language-theoretic algebra. 

Let $\Phi = (\mathsf{A}, \mathsf{T})$ where $\mathsf{A} \subseteq \mathsf{B}$ is a finite set of Boolean variables and $\mathsf{T} \subseteq \mathsf{P}$ is a finite set of program variables. We assume that $\mathsf{A} = \{ \mathtt{b}_1, \ldots, \mathtt{b}_n \}$ is ordered in some fixed but arbitrary way. Let $\Box$ be a new symbol (i.e.~$\Box \notin \bigcup \Sigma$). A \emph{weighted guarded string} over $\Phi$ is a string of the form $$A_0 p_1 A_1 \ldots p_n A_n \Box^{m} \, ,$$ where all $A_i$ are atoms over $\mathsf{A}$ and all $p_i \in \mathsf{T}$. That is, weighted guarded strings over $\Phi$ are a guarded strings over $\Phi$ followed by a string of $m$ copies of $\Box$. We'll write just $s(m)$, where $s \in GS$ and $m \in \mathbb{N}$. A weighted atom is $s(m)$ where $s \in 1_{\mathsf{A}}$.

Intuitively, a guarded string represents an execution trace of a program and the weight $\Box^{m}$ represents the weight of the trace. Hence, (some) sets of weighted guarded strings can be seen as representing the possible execution traces of weighted regular programs. If a set $X$ represents a program and $s(m) \notin X$ for all $m$, then $s$ is not a possible execution trace of the program, i.e.~it has ``infinite weight''.

A set $X$ of weighted guarded strings is \emph{unambiguous} iff $s(n) \in X$ and $s(m) \in X$ only if $n = m$; moreover, $X$ is \emph{crisp} iff $s(m) \in X$ only if $m = 0$, and $X$ is \emph{uniform} only if  $s(n), t(m) \in X$ only if $n = m$. A set $X$ of weighted atoms is \emph{universal} iff for all $A$ there is $n$ such that $A(n) \in X$. Note that a crisp set of weighted guarded strings is just a set of guarded strings, and that uniform sets are unambiguous.

The \emph{unambiguous union} of two sets of weighted guarded strings $X, Y$ is
\begin{equation*}
X \Cup Y := \{ s(\mathrm{min}(U)) \mid U = \{ n \mid s(n) \in X \cup Y \} \} \, .
\end{equation*}
Note that $\Cup$ is not necessarily idempotent, but it is a semilattice join operation on the set of all unambiguous sets of weighted guarded strings. We denote $X \Cup X$ as $X^{\Cup}$.

Coalesced product of weighted guarded strings is defined as follows:
\begin{equation*}
xA(n) \diamond A'y(m) =
\begin{cases}
xAy(n+m) & \text{if } A = A'\\
\text{undefined} & \text{otherwise.}
\end{cases}
\end{equation*}
 
\begin{definition}
Fix a $\Phi$. Let $$ \mathbf{GT} = (K, B, S,\, \cdot, \Cup, \,^{*}, \,^{-}, 1_{\mathsf{A}}, \emptyset)$$ be an algebra such that:
\begin{itemize}
\item $K$ is the set of all unambiguous sets of weighted guarded strings over $\Phi$;
\item $B$ is the set of crisp sets of weighted atoms (i.e.~the set of all sets of atoms);
\item $S$ is the set of uniform universal sets of weighted atoms;
\item $X \cdot Y = (X \diamond Y)^{\Cup}$;
\item $X^{*} = \Bigcup_{n \geq 0} X^{n}$ (where $X^{n}$ is defined using $\cdot$);
\item $\,^{-}$ is complementation on $1_{\mathsf{A}}$.
\end{itemize}
\end{definition}
\noindent (The ordinary lifting of $\diamond$ to sets of weighted guarded strings is not necessarily an unambiguous set.)

\begin{theorem}
$\mathbf{T^{G}} \cong \mathbf{GT}$.
\end{theorem}

\begin{definition}
A valuation in $\mathbf{GT}$ is called \emph{canonical} iff
\begin{itemize}
\item if $\mathtt{p} \in \mathsf{T}$, then $v(\mathtt{p}) = \{ A\mathtt{\,p\,}B \mid A, B \in 1_{\mathsf{A}} \}$;
\item if $\mathtt{b} \in \mathsf{A}$, then $v(\mathtt{b}) = \{ A \mid \neg \mathtt{b} \text{ does not occur in } A \}$.
\end{itemize}
\end{definition}
  
\section{Using Kleene algebra to reason about weighted programs}\label{sec:using}
 
 Using the (quasi)equational theory of Kleene algebra with weights and tests, we can reason about equivalence of weighted programs in a simple algebraic setting.
 
\begin{example}
Recall the ski rental program and Example \ref{exam:ski}. Using Kleene algebra (see \cite{Kozen1994}, Proposition 2.7), we can show that (\ref{eq:ski}) is equivalent to
\begin{equation}\label{eq:ski_star}
\begin{gathered}
(\{ \texttt{neq0}  \} \, \texttt{sub1}  \, \one)^{*} \\
\big ( \{ \texttt{neq0}  \} \, \texttt{sub1}  \, \skis \, \texttt{end}  \, (\{ \texttt{neq0}  \} \, \texttt{sub1}  \, \one)^{*} \  \big)^{*} \\
\{ \neg \texttt{neq0}  \}
\end{gathered}
\end{equation}
Now consider the following equations:
\begin{gather}
\texttt{sub1}^{n} \, \{ \texttt{neq0} \}  = \mathtt{0} 
\label{eq:ski_assumption1}\\
\texttt{end} \, \{ \texttt{neq0} \}  = \mathtt{0}
\label{eq:ski_assumption2}
\end{gather}
Both correspond to reasonable assumptions in the version of the ski rental scenario where the length of the trip in $n$ days, given the intended interpretation of the variables occurring in the equations: if you do subtract $1$ from $n$ $n$-times, then the test $n > 0$ evaluates to False (\ref{eq:ski_assumption1}); and if you assign $n := 0$, then then the test $n > 0$ evaluates to False (\ref{eq:ski_assumption2}). 

Let us define $\gamma^{+n} :=  1 + \gamma + \gamma^{2} + \ldots + \gamma^{n}$. 
%
 It is possible to show that, in each $^{*}$-continuous KAWT where (\ref{eq:ski_assumption1}) and (\ref{eq:ski_assumption2}) hold, and where $1$ is the top element of $S$ (such as algebras where $S$ is isomorphic to the tropical semiring, for example), the program (\ref{eq:ski_star}) is equivalent to 
\begin{equation}\label{eq:ski_plus}
\begin{gathered}
(\{ \texttt{neq0}  \} \, \texttt{sub1}  \, \one)^{+n} \\
\big ( \mathtt{1} + \{ \texttt{neq0}  \} \, \texttt{sub1}  \, \skis \, \texttt{end}  \big) \{ \neg \texttt{neq0}  \}
\end{gathered}
\end{equation}
%
 Hence, it is possible to show using Kleene algebra with weights and tests that, on each input $n$, the program (\ref{eq:ski_star}) is equivalent to a simpler program (\ref{eq:ski_plus}) that does not involve $^{*}$. Hence, in a sense, (\ref{eq:ski_plus}) is \emph{finite}.
\end{example}

\begin{example}
Take $\mathbf{GT}$ over $\Phi = (\mathsf{A}, \mathsf{T})$ where $\mathsf{A} = \{  \{ \texttt{neq0} \}, \{ \neg\texttt{neq0} \} \}$ and $\mathsf{T} = \{ \texttt{sub1}, \texttt{end} \}$. Take any canonical valuation where $v(\one) = \{ A(1) \mid A \in 1_{\mathsf{A}} \} = \{ \{ \texttt{neq0} \}\Box, \{ \neg \texttt{neq0} \}\Box \}$ and $v(\skis) = \{ A(y) \mid A \in 1_{\mathsf{A}} \} = \{ \{ \texttt{neq0} \}\Box^{y}, \{ \neg \texttt{neq0} \}\Box^{y} \}\}$. Then $v$ maps the program (\ref{eq:ski_plus}) to the set $X$ of weighted guarded strings containing:
\begin{gather*}
\{ \neg\texttt{neq0} \} \big ( 0 \big),\quad \{ \texttt{neq0} \} \texttt{sub1} \{ \texttt{neq0} \} \texttt{end} \{ \neg\texttt{neq0} \}\big ( y \big) \, , \\
\{ \texttt{neq0} \} \texttt{sub1} \{ \neg \texttt{neq0} \} \texttt{end} \{ \neg\texttt{neq0} \}\big ( y \big)\, , \\[1mm]
	\{ \texttt{neq0} \} \texttt{sub1} \{\neg \texttt{neq0} \}\big ( 1 \big)\, ,\\ 
		\{ \texttt{neq0} \} \texttt{sub1} \{ \texttt{neq0} \} \texttt{sub1} \{ \texttt{neq0} \} \texttt{end} \{ \neg\texttt{neq0} \}\big ( 1 + y \big)\, ,\\ 
	\{ \texttt{neq0} \} \texttt{sub1} \{ \texttt{neq0} \} \texttt{sub1} \{ \neg \texttt{neq0} \} \texttt{end} \{ \neg\texttt{neq0} \}\big ( 1 + y \big)\, ,\\
{\small \vdots}\\
 \Big (\{ \texttt{neq0} \} \texttt{sub1} \Big)^{n} \{\neg \texttt{neq0} \}\big ( n \big)\, , \\
		 \Big (\{ \texttt{neq0} \} \texttt{sub1} \Big)^{n} \{ \texttt{neq0} \} \texttt{sub1} \{ \texttt{neq0} \} \texttt{end} \{ \neg\texttt{neq0} \}\big ( n + y \big)\, ,\\ 
	 \Big (\{ \texttt{neq0} \} \texttt{sub1} \Big)^{n} \{ \texttt{neq0} \} \texttt{sub1} \{ \neg \texttt{neq0} \} \texttt{end} \{ \neg\texttt{neq0} \}\big ( n + y \big)
\end{gather*}
Every set $G$ of weighted guarded strings in $\mathbf{GT}$ gives a function $\weight_G$ from $1_{\mathsf{A}}$ to the set of sets of weighted atoms defined by
\begin{equation*}
\weight_G (Y) = \{ A(n) \mid \exists x (xA(n) \in Y \diamond G) \}^{\Cup} \,.
\end{equation*}
Hence, $\weight_G (Y)$ is the set weighted atoms representing the \emph{optimal} execution traces from $G$ starting in an atom in $Y$. If $G$ is finite, then $\weight_G$ is obviously computable.

For example, $\weight_X (\{ \neg \texttt{neq0} \}) =\{ \{ \neg \texttt{neq0} \} \} $  (if the ski rental program is run in a state where $n = 0$, then it halts immediately without accumulating any weight), and
\begin{equation*}
 \weight_X (\{\texttt{neq0} \}) =\{ \{ \neg \texttt{neq0} \}\big ( \mathrm{min}( n, y ) \big) \} \, .
 \end{equation*} 
 That is, an optimal run of (\ref{eq:ski_plus}) from a state where $n \neq 0$ will have weight $\mathrm{min}( n, y )$. This result agrees with intuition and the calculation in \cite{BatzEtAl2022} using the weakest preweighting operator.  
\end{example}

\section{Conclusion}\label{sec:conclusion}
We introduced Kleene algebra with weights and tests, an expansion of Kleene algebra with tests suitable for formalizing reasoning about a simplified version of weighted programs discussed in \cite{BatzEtAl2022}. We described constructions of some ``concrete'' KAWT, and we demonstrated by means of an example that KAWT can be used for reasoning about equivalence and optimal runs of weighted programs. 

Many interesting topics need to be left to future research, including a study of free Kleene algebras with weights and tests, questions of decidability and computational complexity, and a systematic accommodation of the weakest preweighting operator of \cite{BatzEtAl2022} into our framework.


\appendix

This appendix contains proofs of some of the technical results stated in the main text, and a lemma that will be used in the proof of Theorem 1.

\begin{lemma}\label{lem:psi}
The following hold in each partial semigroup with identity:
\begin{enumerate}
\item $I(x)$ only if $D(x,x)$ and $x = xx$;
\item $I(x)$, $I(y)$ and $D(x,y)$ only if $I(xy)$.
\end{enumerate}
\end{lemma}
\begin{proof}
1) For all $x$ there is $y$ such that $D(x,y)$ and $I(y)$. Then $x = xy$. But if also $I(x)$, then $y = xy$. Hence, $x = y$ and so $D(x,x)$ and $x = xx$.

2) If $D(x,y)$ and $I(y)$, then $x = xy$. If also $I(x)$, then $I(xy)$. 
\end{proof}

%

\noindent\textbf{Theorem 1.} {\itshape If $\mathbf{P}$ is a partial semigroup with identity and $\mathbf{S}$ is a complete idempotent semiring, then $\mathbf{S^P}$ is a $^{*}$-continuous Kleene algebra with weights and tests where $$(W, \cdot, +, 1, 0) \cong \mathbf{S} \, .$$}

\noindent \textit{Proof.} 
$(S^{G}, +, 0)$ is clearly a commutative monoid. $(S^{G}, \cdot, 1)$ is a monoid:
\begin{itemize}
\item[(i)] $(\lambda_0 \cdot (\lambda_1 \cdot \lambda_2))(x) = \sum_{y,u} \Big\{ \lambda_0(y) \cdot^{\mathbf{S}} \sum_{z,w} \big\{ \lambda_1 (z) \cdot \lambda_2(w) \mid u = z w \And D(z,w) \big\} \:\big| \: x = y u \And D(y,u) \Big\} =$\\
$= \sum_{y,z,w} \big \{ \lambda_0 (y) \cdot^{\mathbf{S}} \big( \lambda_1 (z) \cdot^{\mathbf{S}} \lambda_2 (w) \big) \mid D(y,zw) \And D(z,w) \And x = y(zw) \big \}$\\
$= \sum_{y,z,w}\big \{ \big ( \lambda_0 (y) \cdot^{\mathbf{S}} \lambda_1 (z) \big) \cdot^{\mathbf{S}} \lambda_2 (w) \mid D(y,zw) \And D(z,w) \And x = y(zw) \big \}$ \\
$= \sum_{y,z,w} \big \{ \big ( \lambda_0 (y) \cdot^{\mathbf{S}} \lambda_1 (z) \big) \cdot^{\mathbf{S}} \lambda_2 (w) \mid D(y,z) \And D(yz,w) \And x = (yz)w \big \}$\\
$ = \sum_{y,z,v,w} \Big \{ \sum \big\{ \lambda_0(y) \cdot^{\mathbf{S}} \lambda_1(z) \: \big| \: v = y z \And D(y,z) \big\} \cdot^{\mathbf{S}} \lambda_2(w) \: \big|\: D(v,w) \And x = v w\Big \}$\\
$= ((\lambda_0 \cdot \lambda_1) \cdot \lambda_2)(x)$.

\item[(ii)] $(\lambda \cdot 1)(x) = \sum_{y,z} \{ \lambda (y) \cdot^{\mathbf{S}} 1(z) \mid D(y,z) \And x = y z \}$\\
$= \sum_{y,z} \{ \lambda (y) \mid D(y,z) \And x = yz  \And I(z)\}$\\
$= \sum_{y} \{ \lambda (y) \mid x = y\} = \lambda (x)$.
\item[] ($(1 \cdot \lambda)(x) = \lambda(x)$ is established similarly.)
\end{itemize}
The fourth equality in (i) follows from the definition of a partial semigroup with identity (the first two conditions).\footnote{We note that the equality cannot be established using the weaker definition of a partial semigroup of \cite{GudderSchelp1970}; hence our strengthening.}
 The third equality in (ii) is established using the definition of a partial semigroup with identity as follows. Left to right: If $D(y,z)$ and $I(z)$, then $yz = y$; so if also $x = yz$, then $x = y$. Right to left: for all $y$ there is $z$ such that $D(y,z)$ and $I(z)$, which means that there is $z$ such that $D(y,z)$ and $I(z)$ and $yz= y$. Hence, if $x = y$, then there is $z$ such that $D(y,z)$, $I(z)$ and $x = yz$.

Next we show that $\cdot$ distributes over $+$:
\begin{itemize}
\item[(iii)] $\big( \lambda_0 \cdot (\lambda_1 + \lambda_2) \big)(x)$\\
$ = \sum_{y,z} \big \{ \lambda_0(y) \cdot^{\mathbf{S}} (\lambda_1 + \lambda_2)(z) \:\big|\: D(y,z) \And x = yz  \big \}$\\ 
$ = \sum_{y,z} \big \{ \lambda_0(y) \cdot^{\mathbf{S}} (\lambda_1(z) +^{\mathbf{S}} \lambda_2(z)) \:\big|\: D(y,z) \And x = yz  \big \}$\\ 
$ = \sum_{y,z} \big \{ \big (\lambda_0(y) \cdot^{\mathbf{S}} \lambda_1(z) \big) +^{\mathbf{S}} \big(\lambda_0(y) \cdot^{\mathbf{S}} \lambda_2(z) \big) \:\big|\: D(y,z) \And x = yz  \big \}$\\
$= \sum_{y,z} \{ \lambda_0(y) \cdot^{\mathbf{S}} \lambda_1(z) \mid D(y,z) \And x = yz \} +^{\mathbf{S}}$\\ \mbox{}\hfill
$ \sum_{y,z} \{ \lambda_0(y) \cdot^{\mathbf{S}} \lambda_2(z) \mid D(y,z) \And x = yz \}$\\
$= (\lambda_0 \cdot \lambda_1)(x) +^{\mathbf{S}} (\lambda_0 \cdot \lambda_2)(x)$\\
$= \big ( (\lambda_0 \cdot \lambda_1) + (\lambda_0 \cdot \lambda_2) \big) (x)$
\item[] ($\big( (\lambda_0 + \lambda_1) \cdot \lambda_2 \big)(x) = \big ( (\lambda_0 \cdot \lambda_2) + (\lambda_1 \cdot \lambda_2)\big)(x)$ is established similarly.)
\end{itemize}

To prove that $0$ is the annihilator element it is sufficient to show that $(\lambda \cdot 0)(x) = 0^{\mathbf{S}} = (0 \cdot \lambda)(x)$ for all $x \in G$:
\begin{itemize}
\item[(iv)] $(\lambda \cdot 0)(x) = \sum_{y,z} \{ \lambda(y) \cdot^{\mathbf{S}} 0(z) \mid D(y,z) \And x = yz\}$\\
$= 0^{\mathbf{S}} = $\\
$\sum_{y,z} \{ 0(y) \cdot^{\mathbf{S}} \lambda(z) \mid D(y,z) \And x = yz\} = (0 \cdot \lambda)(x)$.
\end{itemize}

This proves that $\mathbf{S^{P}}$ is an idempotent semiring. To prove that it is also a $^{*}$-continuous Kleene algebra, it is sufficient to show that it satisfies the $^{*}$-continuity condition:
\begin{equation}
 \delta \lambda^{*} \theta = \sum_{n \geq 0} \delta \lambda^{n} \theta
 \end{equation} 
 for all $\delta, \lambda, \theta \in S^{G}$. It is an easy exercise to show that the Kleene star (quasi)equations follow from $^{*}$-continuity. We reason as follows:
 \begin{gather*}
 (\delta\lambda^{*}\theta) (x) = \sum_{y,z} \Big \{ \delta (y) \cdot^{\mathbf{S}} (\lambda^{*} \theta)(z) \:\big|\: D(y,z) \And yz = x \Big \}\\
 = \sum_{y,z,u,v} \Big \{ \delta (y) \cdot^{\mathbf{S}} \big(\lambda^{*}(u) \cdot^{\mathbf{S}} \theta(v) \big) \:\big|\: \\ D(y,z) \And yz = x \And D(u,v) \And uv = z\Big \}\\
 = \sum_{y,z,u,v} \Big \{ \delta (y) \cdot^{\mathbf{S}} \Big( \big (\sum_{n \geq 0} \lambda^{n} (u) \big) \cdot^{\mathbf{S}} \theta(v) \Big) \:\big|\: \\ D(y,z) \And yz = x \And D(u,v) \And uv = z\Big \}\\
  = \sum_{n \geq 0} \sum_{y,z,u,v} \Big \{ \delta (y) \cdot^{\mathbf{S}} \Big( \lambda^{n} (u) \cdot^{\mathbf{S}} \theta(v) \Big) \:\big|\: \\ D(y,z) \And yz = x \And D(u,v) \And uv = z\Big \}\\
    = \sum_{n \geq 0} \sum_{y,z,u,v} \Big \{ \Big (\delta (y) \cdot^{\mathbf{S}}  \lambda^{n} (u) \Big) \cdot^{\mathbf{S}} \theta(v) \:\big|\: \\ D(y,u) \And yu = w \And D(w,v) \And wv = x\Big \}\\
    = \sum_{n \geq 0} \sum_{w,v} \Big \{ \big (\delta \lambda^{n}\big)(w) \cdot^{\mathbf{S}} \theta(v) \:\big|\: D(w,v) \And wv = x\Big \}\\
    = \sum_{n \geq 0} \Big \{ \big (\delta \lambda^{n} \theta \big) (x) \Big \} 
    = \Big ( \sum_{n \geq 0} \delta \lambda^{n} \theta \Big )(x)
 \end{gather*}
The fourth equality holds since $\mathbf{S}$ is a complete semiring. The fifth equality holds thanks to the definition of a partial semigroup with identity (first two conditions).

 Hence, $\mathbf{S^{P}}$ is a $^{*}$-continuous Kleene algebra. To show that it is a $^{*}$-continuous Kleene algebra with tests, we have to show that $B$ is a Boolean algebra and $\,^{-}$ is complementation on $B$. But this follows easily from the definition: $B$ can be equivalently seen as the power set of $I$ (hence clearly a Boolean algebra), and $\,^{-}$ is obviously defined as complementation on $B$.
 
 In order to show that $\mathbf{S^{P}}$ is a Kleene algebra with weights and tests, we have to show that $W$ is closed under the semiring operations $\cdot$ and $+$, and that $0, 1 \in W$. $W$ is the set of functions that assign $0^{\mathbf{S}}$ to elements outside $I$, and that are constant on $I$. Let us denote the set of such functions as $C$. Both $1$ and $0$ are in $C$, and $C$ is clearly closed under $+$. To show that $C$ is are closed under $\cdot$ as well, we reason as follows. Assume that $\lambda, \lambda' \in C$. First we prove that if $x \in I$, then
 \begin{equation}\label{eq:constant}
 (\lambda \cdot \lambda')(x) = \lambda (x) \cdot^{\mathbf{S}} \lambda'(x)
 \end{equation}
Indeed, 
 \begin{gather*}
 \sum_{y,z} \big \{ \lambda(y) \cdot^{\mathbf{S}} \lambda'(z) \mid D(y,z) \And x = yz \big \} \\
 = \lambda(x) \cdot^{\mathbf{S}} \lambda'(x) 
 \end{gather*}
since
  \begin{gather*}
 \sum_{y,z} \big \{ \lambda(y) \cdot^{\mathbf{S}} \lambda'(z) \mid D(y,z) \And x = yz \big \} \\
=   \sum_{y,z} \big \{ \lambda(y) \cdot^{\mathbf{S}} \lambda'(z) \mid D(y,z) \And x = yz \And I(y) \And I(z) \big \}
 \end{gather*}
 (we may forget about $y,z \not\in I$ since $\lambda, \lambda'$ map them to $0^{\mathbf{S}}$) and 
   \begin{gather*}
\sum_{y,z} \big \{ \lambda(y) \cdot^{\mathbf{S}} \lambda'(z) \mid D(y,z) \And x = yz \And I(y) \And I(z) \big \}\\
= \lambda(x) \cdot^{\mathbf{S}} \lambda'(x) \, .
 \end{gather*}
 The latter holds since the set over which the sum is formed contains at least $\lambda(x) \cdot^{\mathbf{S}} \lambda'(x)$ (Lemma \ref{lem:psi}, part 1) and it contains at most $\lambda(x) \cdot^{\mathbf{S}} \lambda'(x)$ since $\lambda, \lambda'$ are constant on $I$. It follows from (\ref{eq:constant}) that $\lambda \cdot \lambda'$ is constant on $I$ since both $\lambda$ and $\lambda'$ are constant on $I$.
 
 Second, we show that $(\lambda \cdot \lambda')(x) = 0^{\mathbf{S}}$ if $x \notin I$. This follows from Lemma \ref{lem:psi}, part 2: if $x \not\in I$, then $x = yz$ and $D(y,z)$ only if $y \notin I$ or $z \notin I$. Hence, if $x \notin I$, then $(\lambda \cdot^{\mathbf{S}} \lambda')(x) = 0^{\mathbf{S}}$. Hence, $(\lambda \cdot \lambda') \in C$ if $\lambda, \lambda' \in C$.
 
 It remains to establish that $(W, \cdot, +, 1, 0)$ is isomorphic to $\mathbf{S}$. Fix and arbitrary $i \in I$ (note that $I \neq \emptyset$ in all partial semigroups with identity) and define $\phi : W \to \mathbf{S}$:
 \begin{equation*}
 \phi (\lambda) = \lambda(i)
 \end{equation*}
 The mapping $\phi$ is a bijective homomorphism. Homomorphism: $\phi (1) = 1(i) = 1^{\mathbf{S}}$; $\phi (0) = 0(i) = 0^{\mathbf{S}}$; $\phi (\lambda + \lambda') = (\lambda + \lambda')(i) = \lambda(i) +^{\mathbf{S}} \lambda'(i) = \phi(\lambda) + ^{\mathbf{S}} \phi(\lambda')$; $\phi(\lambda \cdot \lambda') = (\lambda \cdot \lambda')(i) = \lambda(i) \cdot^{\mathbf{S}} \lambda'(i)$ by (\ref{eq:constant}) $= \phi(\lambda) \cdot^{\mathbf{S}} \phi(\lambda')$. Surjective: $W$ is the set of \emph{all} functions that are constant on $I$ and assign $0^{\mathbf{S}}$ to elements $x \notin I$. Injective: if $\phi(\lambda) = \phi(\lambda')$, then $\lambda(i) = \lambda'(i)$, and then $\lambda = \lambda'$ since $\lambda$ and $\lambda'$ are assumed constant on $I$.
\qed

\

\noindent\textbf{Theorem 2.} $\mathbf{T^{G}} \cong \mathbf{GT}$.

\smallskip

\noindent\textit{Proof.} Define $\tau : (\mathbb{N}^{\infty})^{GS} \to K^{\mathbf{GT}}$ such that
\begin{equation*}
\tau (\lambda) = \{ s (n) \mid \lambda (s) = n \And n \neq \infty \} \, .
\end{equation*}

The function $\tau$ is clearly a bijection between $(\mathbb{N}^{\infty})^{GS}$ and $K^{\mathbf{GT}}$. Moreover,
\begin{itemize}
\item[(i)] $B^{\mathbf{GT}} = \{ \tau (\lambda) \mid \lambda \in B\}$ since $\lambda \in B$ iff $\lambda : GT \to \{ 0^{\mathbb{N}}, \infty \}$ and $\lambda (s) = 0^{\mathbb{N}}$ only if $s \in 1_{\mathsf{A}}$ iff $\tau (\lambda)$ is a crisp set of weighted atoms;

\item[(ii)] $S^{\mathbf{GT}} = \{ \tau (\lambda) \mid \lambda \in S \}$ since $\lambda \in S$ iff $\lambda (s) = \infty$ for $s \notin 1_{\mathsf{A}}$ and $\lambda$ is constant on $1_{\mathsf{A}}$ iff  $\tau(\lambda)$ is an uniform set of weighted atoms.
\end{itemize}
Next we need to show that $\tau$ is a homomorphism:
\begin{itemize}
\item[(iii)] $\tau (1) = \{ s(0) \mid s \in 1_{\mathsf{A}} \} = 1^{\mathbf{GT}}$;

\item[(iv)] $\tau (0) = \emptyset = 0^{\mathbf{GT}}$;

\item[(v)] $s(n) \in \tau (\lambda \cdot \lambda')$
$\iff$ $(\lambda \cdot \lambda')(s) = n \neq \infty$\\
$\iff$ $n \neq \infty$ and $n = \mathrm{min} \{ \lambda (t) +^{\mathbb{N}^{\infty}} \lambda'(u) \}$ for $t, u \in GS$ such that $t \diamond u = s$\\
$\iff$ $\exists t, u \in GS : s = t \diamond u$ and $n = \lambda (t) +^{\mathbb{N}^{\infty}} \lambda'(u)$ and $\lambda(t) \neq \infty$ and $\lambda'(u) \neq \infty$ and \\ \mbox{}\hfill 
	$\forall t', u' \in GS (s = t' \diamond u' \to n \leq^{\mathbb{N}^{\infty}} \lambda(t') + \lambda'(u'))$\\
$\iff$ $s(n) \in \tau(\lambda) \diamond \tau (\lambda')$ and \\ \mbox{}\hfill 
	$\forall t', u' \in GS (s = t' \diamond u' \to n \leq^{\mathbb{N}^{\infty}} \lambda(t') + \lambda'(u'))$\\
$\iff$ $s(n) \in \tau(\lambda) \cdot^{\mathbf{GT}} \tau(\lambda')$;	

\item[(vi)] $s(n) \in \tau (\lambda + \lambda')$ $\iff$ $(\lambda + \lambda')(s) = n \neq \infty$\\
$\iff$ $n = \mathrm{min} \{ \lambda(s), \lambda'(s) \}$ and $n \neq \infty$\\
$\iff$ $s(n) \in \tau (\lambda) \cup \tau(\lambda')$ and\\ \mbox{}\hfill
		$\forall m \in \mathbb{N} (s(m) \in \tau (\lambda) \cup \tau (\lambda') \to n \leq^{\mathbb{N}^{\infty}} m)$\\
$\iff$ $s(n) \in \tau(\lambda) \Cup \tau(\lambda')$

\item[(vii)] $s(n) \in \tau (\lambda^{*})$ $\iff$ $n = \mathrm{min}_{m \in \mathbb{N}} \{ \lambda^{m}(s) \}$\\
$\iff$ $\exists m \in \mathbb{N}: n = \underbrace{(\lambda \cdot \ldots \cdot \lambda)}_{m\text{-times}} (s)$ and \\ \mbox{}\hfill
		$\forall k \big ( n \leq^{\mathbf{\mathbb{N}}} \underbrace{(\lambda \cdot \ldots \cdot \lambda)}_{k\text{-times}} (s) \big)$\\
$\iff$ $s(n) \in \bigcup_{m \in \mathbb{N}} \tau(\lambda)^{m}$ and\\ \mbox{}\hfill  $n = \mathrm{min} \big\{ n' \: \big| \: s(n') \in \bigcup_{k \in \mathbb{N}} \tau(\lambda)^{k} \big\}$\\
$\iff$ $s(n) \in \Big ( \bigcup_{m \in \mathbb{N}} \tau(\lambda)^{m} \Big)^{\Cup} = \Bigcup_{m \in \mathbb{N}} \tau(\lambda)^{m}$\\
$\iff$ $s(n) \in \big ( \tau (\lambda)\big)^{*^{\mathbf{GT}}}$;		

\item[(viii)] for $\lambda \in B$: $s(0^{\mathbb{N}}) \in \tau (\overline{\lambda})$\\ $\iff$ $\overline{\lambda}(s) = 0^{\mathbb{N}}$ $\iff$ $\lambda(s) = \infty$\\
$\iff$ $s(0^{\mathbb{N}}) \notin \tau(\lambda)$ $\iff$ $s(0^{\mathbb{N}}) \in \overline{\tau(\lambda)}$.
\end{itemize}
\qed

Note that we didn't need to assume in the proof that $\mathbf{GT} \in \mathsf{KAWT}$, but this follows from Theorem 2.

\

\textbf{Example 7.} {\itshape If $\mathbf{K} \in \mathsf{KAWT}$ such that (\ref{eq:ski_assumption1}) and (\ref{eq:ski_assumption2}) holds in $\mathbf{K}$, and $\mathtt{1}$ is the top element of $S^{\mathbf{K}}$, then (\ref{eq:ski_star}) is equivalent to (\ref{eq:ski_plus}).}

\textit{Proof.} By $^{*}$-continuity and (\ref{eq:ski_assumption2}), 
\begin{equation*}
\{ \texttt{neq0}  \} \, \texttt{sub1}  \, \skis \, \texttt{end}  \, (\{ \texttt{neq0}  \} \, \texttt{sub1}  \, \one)^{*} 
\end{equation*}
is equivalent to
\begin{equation*}
\{ \texttt{neq0}  \} \, \texttt{sub1}  \, \skis \, \texttt{end}
\end{equation*}
and so the second line of (\ref{eq:ski_star}) is equivalent to 
\begin{equation*}
\mathtt{1} + \{ \texttt{neq0}  \} \, \texttt{sub1}  \, \skis \, \texttt{end} \, .
\end{equation*}
By $^{*}$-continuity, 
\begin{equation*}
(\{ \texttt{neq0}  \} \, \texttt{sub1}  \, \one)^{*}
\end{equation*}
is equivalent to
\begin{equation}\label{eq:ski_sum}
\sum_{n \geq 0} \big ( \{ \texttt{neq0}  \} \, \texttt{sub1}  \, \one \big )^{n}
\end{equation}
However, since $\{ \texttt{neq0}  \}, \one \leq \mathtt{1}$,
\begin{equation*}
\big ( \{ \texttt{neq0}  \} \, \texttt{sub1}  \, \one \big )^{n+1}
\end{equation*}
is less or equal to
\begin{equation*}
 \texttt{sub1}^{n} \big ( \{ \texttt{neq0}  \} \, \texttt{sub1}  \, \one \big) \, ,
 \end{equation*} 
 which equals $\mathtt{0}$ by (\ref{eq:ski_assumption1}). Hence, (\ref{eq:ski_sum}) is equivalent to
 \begin{equation*}
 \big ( \{ \texttt{neq0}  \} \, \texttt{sub1}  \, \one \big )^{+n} \, .
 \end{equation*}
 \qed

\end{document}